\documentclass[12pt,english]{article}
\usepackage{lmodern}
\usepackage{lmodern}
\usepackage[T1]{fontenc}
\usepackage[cp1252]{inputenc}
\usepackage{geometry}
\usepackage{color}
\usepackage{babel}
\usepackage{amsmath}
\usepackage{amsthm}
\usepackage{amssymb}
\usepackage{setspace}
\usepackage[authoryear,comma,longnamesfirst]{natbib}
\usepackage{microtype}
\usepackage[unicode=true,
 bookmarks=false,
 breaklinks=false,pdfborder={0 0 1},backref=section,colorlinks=true]
 {hyperref}
\hypersetup{pdftitle={"RRDEU"},
 pdfauthor={"AguiarKashaev"},
 pdfnewwindow=true,pdfstartview=FitH,urlcolor=blue!90!red!45!black,citecolor=blue!90!red!45!black,linkcolor=red!90!black}

\makeatletter
\theoremstyle{plain}

\theoremstyle{definition}
\newtheorem{defn}{\protect\definitionname}
\theoremstyle{definition}
\newtheorem{example}{\protect\examplename}
\theoremstyle{plain}
\newtheorem{lem}{\protect\lemmaname}

\usepackage{amsfonts}% AMS math
\usepackage{dsfont}\usepackage{mathrsfs}\usepackage{ushort}% Math style

\usepackage{titlesec}\usepackage{titling}% Section titles
\usepackage{caption}% Spacing
\usepackage{enumitem}\usepackage{booktabs}% Tables & lists
\usepackage{tikz}
\usepackage{pstricks}\usepackage{pst-all}% Figures
\usepackage{pst-plot}
\usepackage{pst-node}\usepackage{pst-3dplot}\usepackage{sgamevar}% Strategic form games
\usepackage{subcaption}
\usepackage{lscape}
\usepackage{pifont}
\usepackage{longtable}
\usepackage{algorithmic}
\usepackage{pdflscape}
\usepackage{rotating}
\usepackage{float}

\setenumerate{label=\small(\roman*)}
\DeclareCaptionFont{fancy}{\bfseries\sffamily}
\gamemathtrue
\allowdisplaybreaks

\providecommand{\psreset}{\psset{%
		linewidth=0.3pt,linestyle=solid,linecolor=black,
		dotsize=2.5pt,dotsep=2.5pt,arrowsize=4pt,
		fillstyle=none,fillcolor=white,
		showpoints=false,arrows=-,linearc=0,framearc=0,
		hatchsep=2pt,hatchwidth=0.2pt,nodesep=4pt,opacity=1}
	\psset{gridcolor=black!60, subgridcolor=black!30}
}

\psreset

\usepackage{graphicx}
\graphicspath{ {figures/} }
\titleformat{\section}[block]{\centering\large\bfseries\sffamily}{\thesection.}{0.5em}{}
\titleformat{\subsection}[block]{\flushleft\bfseries}{\thesubsection.}{0.5em}{}
\titleformat{\subsection}[block]{\flushleft\bfseries\sffamily}{\thesubsection.}{0.5em}{}
\titleformat{\subsubsection}[runin]{\normalsize\bfseries\sffamily}{\bfseries\upshape\sffamily\thesubsubsection.}{0.5em}{}[.--\:]
\renewcommand{\thesubsubsection}{\arabic{section}.\arabic{subsection}.\arabic{subsubsection}}
\titlespacing{\section}{0ex}{10ex}{5ex}
\titlespacing{\subsection}{0in}{6ex}{3ex}
\titlespacing{\subsubsection}{0mm}{2ex}{0.5em}
\pretitle{\begin{center}\LARGE\bfseries\sffamily}
\posttitle{\par\end{center}\vskip 0.5em}
\preauthor{\begin{center} \large \lineskip 0.5em\begin{tabular}[t]{c}}
\postauthor{\end{tabular}\par\end{center}}
\predate{\begin{center}\small}
\postdate{\par\end{center}}
\providecommand{\abstitle}[1]{{\par\vspace*{2ex}\small\bfseries\sffamily #1}\hspace*{1ex}}
\renewenvironment{abstract}%
{\begin{center}\begin{minipage}{0.8\linewidth}%
			\abstitle{Abstract}\small}%
		{\end{minipage}\end{center}\vfill\clearpage}

\providecommand{\Char}[1]{\mathds{1}\left(\,#1\,\right)}
\providecommand{\Real}{{\mathds{R}}}

\providecommand{\tr}{^{\prime}}

\providecommand{\rand}[1]{\mathbf{#1}}

\providecommand{\norm}[1]{\left\lVert#1\right\rVert}

\providecommand{\abs}[1]{\left\lvert#1\right\rvert}

\newcommand{\A}{A}

\theoremstyle{remark}
  \theoremstyle{plain}
  \theoremstyle{definition}
    \newtheorem{proposition}{\protect\propositionname}\theoremstyle{definition}
  \theoremstyle{plain}
\newtheorem{theorem}{\protect\theoremname}\theoremstyle{plain}
  \theoremstyle{definition}
  \providecommand{\assumptionname}{Assumption}

  \providecommand{\definitionname}{Definition}
  \providecommand{\lemmaname}{Lemma}
  \providecommand{\propositionname}{Proposition}
  \providecommand{\remarkname}{Remark}
\providecommand{\corollaryname}{Corollary}
\providecommand{\theoremname}{Theorem}
\providecommand{\examplename}{Example}

\makeatother

\providecommand{\definitionname}{Definition}
\providecommand{\examplename}{Example}
\providecommand{\lemmaname}{Lemma}
\providecommand{\theoremname}{Theorem}

\begin{document}
\title{Random Rank-Dependent Expected Utility\thanks{We thank Maria Jose Boccardi and Jeongbin Kim for allowing the use of the experimental dataset from \citet{ABKK2021}.}\thanks{The ``\textcircled{r}'' symbol indicates that the authors' names are in certified random order, as described by \citet{ray2018certified}. We gratefully acknowledge financial support from 
Social Sciences and Humanities Research Council Insight Development Grant.}}
\author{ 
	Nail Kashaev \textcircled{r}
	Victor H. Aguiar\thanks{Kashaev: Department of Economics, University of Western Ontario; \href{mailto:nkashaev@uwo.ca}{nkashaev@uwo.ca}. Aguiar: Department of Economics, University of Western Ontario; \href{mailto:vaguiar@uwo.ca}{vaguiar@uwo.ca}.}
	}
\date{December, 2021}
\maketitle

\begin{abstract}
We present a novel characterization of random rank-dependent expected
utility for finite datasets and finite prizes. The test lends
itself to statistical testing using the tools in \citet{kitamura2018nonparametric}. 


\end{abstract}

\section{Introduction}
The rank-dependent expected utility (RDEU) model, first proposed by Quiggin \citeyearpar{quiggin1982theory}, is one of the main alternatives to the expected utility (EU) model. RDEU is popular because it extends EU to accommodate known empirical anomalies that violate the predictions of EU (e.g., Allais Paradox) while at the same time it produces sharp comparative statics and preserves transitivity \citep{quiggin1991comparative}.  
\par 
The RDEU model extends EU by allowing cumulative probabilities of lotteries to be weighted by a weighting function. This means that RDEU does not rely on the independence axiom.\footnote{The independence axiom dictates that for two lotteries $p,q$,  $p$ is weakly preferred to $q$ if and only if $\alpha p +(1-\alpha)r$ is weakly preferred to $\alpha q +(1-\alpha) r$ for any lottery $r$, and $alpha\in [0,1]$.} According to Quiggin \citeyearpar{quiggin1991comparative}, RDEU is EU with respect to a transformed probability distribution. Formally, we say a decision maker (DM) that can be described by a RDEU ranks lottery $p$ over lottery $q$ if and only if there exist a (Bernoulli) utility function $u$ and a weighting function $\phi:[0,1]\to [0,1]$, $\phi(0)=0$, $\phi(1)=1$, such that
\[
\sum_{k=1}^K\left[\phi\left(\sum_{t=1}^{k}p_{t}\right)-\phi\left(\sum_{t=1}^{k-1}p_{t}\right)\right]u(x_{k}) > \sum_{k=1}^K\left[\phi\left(\sum_{t=1}^{k}q_{t}\right)-\phi\left(\sum_{t=1}^{k-1}q_{t}\right)\right]u(x_{k}),
\]
where prizes $x_k$ are ranked by some primitive order and we use convention that $\sum_{t=1}^0p_t=0$. 
\par 
Here we study the empirical implications of RDEU when we observe a cross-section of choices, from a finite collection of menus and lotteries, made by DMs that follow the RDEU rule. Since we allow the primitives of RDEU to be heterogeneous among the DMs (e.g., random $\phi$ and random $u$) we call this model of population behavior Random Rank-Dependent Expected Utility (RRDEU). We fully characterize the RRDEU model and provide a statistical test of it using the tools in \citet{kitamura2018nonparametric} (henceforth KS). 
\par 
As a byproduct of our characterization of RRDEU in a finite stochastic choice dataset (limited stochastic dataset), we provide a characterization of random expected utility (REU). The characterization of REU is the first of its type for limited stochastic datasets. The seminal work of \citet{gul2006random} (henceforth GP) characterizing the REU for an infinite stochastic dataset may provide false positives when applied to a limited stochastic dataset. The two main axioms in the GP characterization are regularity and independence. Regularity requires that the probability of choosing a lottery weakly decreases as the number of alternatives in a menu grows. Independence in GP requires that if we expand a menu $A$ by mixing all the lotteries inside it with an outside lottery $r$, then the probability of choice of the mixture lotteries in the new menu is the same as the probability of choice of the original ones. The problem with the characterization in GP is that if we have only $3$ menus that are not ranked by the set inclusion relation and are not related by the mixture operations required by the independence assumption in GP, the axioms are trivially satisfied. However, we will show that even with $3$ menus such as the ones we have described there are still empirical properties of REU that may fail showing that the test in GP fails for limited stochastic choice dataset. 
\par 
\begin{example}
[Counterexample to \citet{gul2006random}]. Consider three lotteries $\{p,q,p'=\alpha p+ (1-\alpha) r\}$, and we observe three menus $\{p,q\}$, $\{p',q\}$, $\{p',p\}$ such that $\rho_{\{p,q\}}(p)=1$, $\rho_{\{p',p\}}(p')>0$ and $\rho_{\{p',q\}}(p')=0$ (where $\rho_A(a)$ denotes the probability of choosing $a$ in menu $A$). We can observe that the axioms of independence and regularity of GP are trivially satisfied. However, for any fixed $\alpha>0$,  it has to be if the DMs can be described by EU behavior, then $\rho_{\{p',q\}}(p')>0$. The reason is that for the given restrictions on $\rho$, there is at least one expected utility function $U$ with positive mass, such that $U(r)>U(p)>U(q)$, which means that $\rho_{\{p',q\}}(p')>0$. This means that this limited stochastic dataset is inconsistent with REU.    
\end{example}
The previous example shows that the test in GP fails when we consider a limited stochastic dataset. Since EU is a special case of RDEU this is also a problem for the latter. Indeed similar issues will appear in GP's style characterizations of RRDEU.  Our results provide a fix for this issue. The intuition behind our test is that for  any linear order on the set of finite lotteries, we can uncover whether there is $u$ and $\phi$ that can describe this order and, hence, make it consistent with RDEU. If the answer to the previous question is affirmative, then we can use that linear order as the support of a random utility rule. We show that this test can be done using a simple quadratic program in the case of RRDEU that simplifies to a linear program in the case of REU. 
\par
Since we have characterized the empirical content of RRDEU that has a natural population interpretation, our tests can be used in experimental and field datasets that have a cross-section of choices with limited menu variation. Finally, since our test lends itself to statistical testing it can deal with sampling variability which is not the case of GP. To the best of our knowledge no other work has provided a characterization and statistical test for RRDEU in a cross-section of choices. The closest work to ours is by \citet{polisson2020revealed} that provides a nonparametric test of EU and RDEU for a setup with linear budgets and a time series of choices from the same individual in contrast to our cross-section setup.  
\par
In Section~\ref{sec:model}, we describe our setup and formally define the RRDEU. In Section~\ref{sec: construction of rankings}, we provide a construction of the set of all linear orders consistent with RDEU for a given finite set of lotteries and establish the main theoretical result. In Section~\ref{secc:test}, we apply our method to an experimental dataset. Finally, we conclude in Section~\ref{sec: conclusion}.

\section{Model\label{sec:model}}
We consider a finite set of distinct alternatives $X=\{x_{k}\}_{k=1}^K\subset\Real$ such that $x_k<x_{k+1}$ for all $k=1,\dots,K-1$. Let $\Delta(X)$ be the set of all lotteries (distributions) defined on $X$ and $\Pi\subseteq\Delta(X)$ be a finite subset of it. Let $\mathcal{A}\subseteq2^{\Pi}\setminus\{\emptyset\}$ be a collection of menus of lotteries. The probabilistic choice rules are $\rho_{A}\in\Delta(A)$, where $A\in\mathcal{A}$ and $\Delta(A)$ is the simplex defined on menu $A$. A stochastic choice dataset is $\rho=(\rho_{A})_{A\in\mathcal{A}}$. The set of all linear orders on $\Delta(X)\times\Delta(X)$ is denoted by $S$.

\begin{defn}[Random Utility, RU] We say that $\rho$ admits a random utility (RU) representation if there exists $\mu\in\Delta(S)$ such that
\[
\rho_{A}(a)=\sum_{\succ\in S}\mu(\succ)\Char{a\succ b,\:\forall b\in A}
\]
for all $A\in\mathcal{A}$ and $a\in A$. 
\end{defn}
Let $\mathcal{U}$ be the set of all utility functions from $X$ to $\Real$, and $\mathcal{P}$ be the set of all weighting functions $\phi:[0,1]\to[0,1]$ such that $\phi(0)=0$ and $\phi(1)=1$. Every $u$ and $\phi$ together define a rank-dependent expected utility function over the set of lotteries $\Delta(X)$ as
\[
U_{u,\phi}(p)=\sum_{k=1}^K\left[\phi\left(\sum_{t=1}^{k}p_{t}\right)-\phi\left(\sum_{t=1}^{k-1}p_{t}\right)\right]u(x_{k})
\]
for any $p\in\Delta(X)$, where we use convention that $\sum_{t=1}^0p_t=0$. The rank-dependent expected utility function coincides with the standard expected utility function if $\phi(x)=x$.
\par
Similar to the case with the expected utility function, the rank-dependent expected utility function allows us to define preference orders that are consistent with rank-dependent expected utility.
\begin{defn}
We say that the linear order $\succ\in S$ is rank-dependent expected utility linear order or $\succ\in R$ if there exists $u\in\mathcal{U}$ and $\phi\in\mathcal{P}$ such that for any $p,q\in\Delta(X)$
\[
p\succ q \iff U_{u,\phi}(p)>U_{u,\phi}(q).
\]
\end{defn}

Given the set of all rank-dependent expected utility orders we can define when observed data $\rho$ could have been generated by a heterogeneous population of DMs with rank-dependent expected utility orders.
\begin{defn}[Random Rank-Dependent Expected Utility, RRDEU] We say that $\rho$ admits a random rank-dependent expected utility (RRDEU) representation if there exists $\mu\in\Delta\left(R\right)$ such that
\[
\rho_{A}(a)=\sum_{\succ\in R}\mu(\succ)\Char{a\succ b,\:\forall b\in A}
\]
for all $A\in\mathcal{A}$ and $a\in A$.
\end{defn}
The definition of RRDEU allows us to test whether a given data is consistent with it the same way we do it for RU. The only difference is that instead of working with all possible strict linear orders $S$, we need to compute the set $R$. 

\section{Construction of the set of rank-dependent expected utility linear orders\label{sec: construction of rankings}}
Before we describe the general procedure for construction of $R$, we explain how the procedure works with EU. 
\subsection{Expected Utility Model}
As we mentioned before, the standard expected utility model is the special case RRDEU with $\phi(x)=x$. Thus, to check whether a given linear order $\succ$ is expected utility linear order we need to check whether there exists a utility function $u$ such that
\[
p\succ q \iff \sum_{k=1}^Kp_ku(x_k)>\sum_{k=1}^Kq_ku(x_k)
\]
for all $p,q\in \Delta(X)$.
Note that, in spirit of revealed preference inequalities, the latter is equivalent to requiring the existence of a set of reals $\{v_k\}_{k=1}^K$ such that 
\[
p\succ q \iff \sum_{k=1}^Kp_kv_k>\sum_{k=1}^Kq_kv_k.
\]
Since the last inequality must hold for all possible $p$ and $q$, we end up having a system of linear inequalities
\[
\sum_{k=1}^K(p_k-q_k)v_k>0, \quad p,q\in \Delta(X).
\]
If there are finitely many lotteries (i.e., we consider $p,q\in\Pi$, $\abs{\Pi}<\infty$), then let $p^{\succ(l)}$ denote the $l$-th best lottery in $\Pi$ according to $\succ$, then it suffices to conduct $\abs{\Pi}-1$ comparisons to get the following system of linear inequalities:
\[
\sum_{k=1}^K(p^{\succ(l)}_k-p^{\succ(l+1)}_k)v_k>0, \quad l=1,\dots,\abs{\Pi}-1.
\]
Checking whether a finite set of linear inequalities has a solution is a linear programming problem and can be done very efficiently and fast.

\subsection{Rank-dependent Expected Utility}
Next we describe how to extend the procedure for the standard expected utility model to the rank-dependent one.
Given a candidate linear order $\succ\in S$, we need to check that there exist $u$ and $\phi$ such that for any lotteries $p$ and $q$
\[
\sum_{k=1}^K\left[\phi\left(\sum_{t=1}^{k}p_{t}\right)-\phi\left(\sum_{t=1}^{k-1}p_{t}\right)-\phi\left(\sum_{t=1}^{k}q_{t}\right)+\phi\left(\sum_{t=1}^{k-1}q_{t}\right)\right]u(x_{k})>0.
\]
For a finite set of lotteries, we can reformulate the problem as requiring existence of two sets of reals $\{v_k\}_{k=1}^K$ and $\{f_{l,k}\}_{k=1,\:l=1}^{K-1,\:\abs{\Pi}}$ such that
\begin{align}\label{eq: main system}
    &\nonumber\sum_{k=1}^K\left[f_{l,k}-f_{l,k-1}-f_{l+1,k}+f_{l+1,k-1}\right]v_k>0,\quad l=1,\dots,\abs{\Pi}-1,\\
    &\nonumber f_{l,k}=f_{s,m}, \text{ whenever } \sum_{t=1}^k p^{\succ(l)}_t=\sum_{t=1}^m p^{\succ(s)}_t,\\
    &f_{l,k}=0, \text{ whenever } \sum_{t=1}^k p^{\succ(l)}_t=0,\\
    &\nonumber f_{l,k}=1, \text{ whenever } \sum_{t=1}^k p^{\succ(l)}_t=1,\\
    &\nonumber 0\leq f_{l,k}\leq1.
\end{align}
The first set of inequalities just imply that the rank-dependent utilities imply the order consistent with $\succ$. The rest of the constrains just follow from the definition of $\phi$ (i.e., $\phi:[0,1]\to[0,1]$, $\phi(0)=0$, and $\phi(1)=1$). Checking whether the system (\ref{eq: main system}) is satisfied for some $\{v_k\}_{k=1}^K$ and $\{f_{l,k}\}_{k=1,\:l=1}^{K-1,\:\abs{\Pi}}$ is a quadratic problem, which also can be solved efficiently and fast.
\par
The next lemma formally establishes that the system (\ref{eq: main system}) provides necessary conditions for $\succ$ being rank-dependent expected utility order.
\begin{lem}\label{lemma: necessety}
Given $\Pi\subseteq\Delta(X)$, $\abs{\Pi}<\infty$, the linear order $\succ\in S$ belongs to $R$ only if there exist $\{v_k\}_{k=1}^K$ and $\{f_{l,k}\}_{k=1,\:l=1}^{K-1,\:\abs{\Pi}}$ such that the system (\ref{eq: main system}) is satisfied.
\end{lem}
\begin{proof}
The result is trivially satisfied if one takes $v_k=u(x_k)$ and $f_{l,k}=\phi\left(\sum_{t=1}^{k}p^{\succ(l)}_{t}\right)$.
\end{proof}
\par
Next, we state the extension proposition that guaranties that the system (\ref{eq: main system}) leads to a rank-dependent expected utility order.
\begin{proposition}\label{prop:sufficiency}
Given $\Pi\subseteq\Delta(X)$, $\abs{\Pi}<\infty$, if there exist $\{v_k\}_{k=1}^K$ and $\{f_{l,k}\}_{k=1,\:l=1}^{K-1,\:\abs{\Pi}}$ such that the system (\ref{eq: main system}) is satisfied for some $\succ^* \in\Pi\times\Pi$, then there exists $\succ\in R$ that coincides with $\succ^*$ on $\Pi\times\Pi$.
\end{proposition}
\begin{proof}
The proposition requires building functions $u$ and $\phi$ from $\{v_k\}_{k=1}^K$ and $\{f_{l,k}\}_{k=1,\:l=1}^{K-1,\:\abs{\Pi}}$. The simplest utility function that will make $\succ$ rank-dependent expected utility order is any peace-wise linear function with nods at points $\{(x_k,v_k)\}_{k=1}^{K}$. To construct $\phi$ one can also take any piece-wise linear function with nods at points $\{(\sum_{t=1}^{k}p^{\succ(l)},f_{k,l})\}_{k=1,\:l=1}^{K-1,\:\abs{\Pi}}$.
\end{proof}
\par
After the set $R$ is constructed, the problem of testing whether a stochastic dataset $\rho$ admits a RRDEU representation can be done by testing the restricted RU model as in McFadden Richter \citet{mcfadden1990stochastic} and KS.
\par 
Let $R_{\Pi}$ be the set of all linear orders on $\Pi\times\Pi$ that consistent with RDE (i.e.,  the restriction of $R$ to $\Pi$).\footnote{The set of linear orders $R_{\Pi}$ can be replaced by the set of expected utility ranking.} We can use system~\ref{eq: main system} to uncover the elements of $R_{\Pi}$. We want to test whether $\rho$ admits a RRDEU representation, this turns out to be equivalent to testing whether $\rho$ can be generated by a population of DMs whose preferences are in $R_{\Pi}$.  Let $B$ be the matrix of the size $d_{\rho}\times \abs{R_{\Pi}}$, where $d_{\rho}$ is the dimensionality of $\rho$ such that $(k,l)$ element of it is equal to
\[
B_{k,l}=\Char{a\in\A}\Char{a\succ_l c,\:\forall\: c\in A},
\]
where $k$ corresponds to a pair $(a,A)$ such that $a\in A$, and $\succ_l$ is $l$-th linear order from $R_{\Pi}$. By \citet{mcfadden1990stochastic} and KS, $\rho$ can be explained by a population of DMs whose preferences are in $R_{\Pi}$ if and only if
\[
\rho=Bv
\]
for some $v\in\Real^{\abs{R_{\Pi}}}_{+}$.
This is our main result:
\begin{theorem}\label{thm:main theorem}
The following are equivalent:
\begin{enumerate}
    \item A stochastic dataset $\rho$ admits a RRDEU representation.
    \item A stochastic dataset $\rho$ is such that there exists some $v\in\Real^{\abs{R_{\Pi}}}_{+}$ such that $\rho=Bv$.\footnote{This result generalizes some informal results about EU in \citet{ABKKold}.}
\end{enumerate}
\end{theorem}
The proof of Theorem~\ref{thm:main theorem} follows from Lemma~\ref{lemma: necessety}, Proposition~\ref{prop:sufficiency}, and KS. 
\par 
The RRDEU model is a strict generalization of REU. It is less general that RU because it predicts first order stochastic dominance (under mild monotonicity constraints) which RU does not require. \citet{quiggin1991comparative} provides additional implications of RDEU. 

\subsection{Shape Restrictions\label{secc:shape constraint}}
Our framework allows us to impose monotonicity or concavity/convexity on $\phi$.
Monotonicity is a normatively desirable property. \citet{abdellaoui2002genuine} shows that convexity of $\phi$ is related to risk aversion. In particular, imposing it will guarantee that consumers are risk averse if the Bernoulli utility is set to the identity. \par
In particular, to impose the restriction that $\phi$ is weakly monotonically increasing it suffices to enhance the system (\ref{eq: main system}) with the following set of linear inequality constraints:
\begin{equation}\label{eq: monotonicity}
f_{l,k}\geq f_{s,m}, \text{ whenever } \sum_{t=1}^k p^{\succ(l)}_t\geq\sum_{t=1}^m p^{\succ(s)}_t.
\end{equation}
\par
Convexity (concavity) can be imposed using cyclical monotonicity \citep{rockafellar2015convex}. In particular, for all cycles of indices $\{l_j,k_j\}_{j=1}^{J}$ such that $l_{J}=l_1$ and $k_{J}=k_1$ let $\{f_{l,k}\}_{k=1,\:l=1}^{K-1,\:\abs{\Pi}}$ satisfy
\begin{equation}\label{eq: convexity}
\sum_{j=1}^{J-1} \left(f_{l_{j+1},k_{j+1}}-f_{l_{j},k_{j}}\right)\sum_{t=1}^{k_{j+1}} p^{\succ(l_{j+1})}_t\geq 0. 
\end{equation}
\begin{proposition}
Given a finite set of lotteries $\Pi$, the linear order $\succ\in S$ is rank-dependent expected utility linear order with weakly increasing and convex $\phi$ if and only if there exist $\{v_k\}_{k=1}^K$ and $\{f_{l,k}\}_{k=1,\:l=1}^{K-1,\:\abs{\Pi}}$ such that the system (\ref{eq: main system}) together with restrictions (\ref{eq: monotonicity}) and (\ref{eq: convexity}) is satisfied.
\end{proposition}

\subsection{Econometric Testing}\label{subsec: econometric testing}
Here we deal with sampling variability. Sampling variability arises from the fact that $\rho$ can be only consistently estimated by the realized choice frequencies $\hat{\rho}$. This section follows closely \citet{ABKK2021}. First we need some additional notation. For every $A\in\mathcal{A}$, let $n_{A}$ denote the number of individuals in the sample that faced choice set $A$, and let $\rand{a}_{i,A}$, $i=1,\dots,n_{A}$, be the observed choice of individual $i$ from choice set $A$. Here we assume that the researcher observes a cross-section of observations for every $A\in\mathcal{A}$. Given this we define the estimated stochastic choice rule as 
\[
\hat \rho=(\hat \rho_A(a))_{A\in\mathcal{A},a\in A}.
\]
with $\hat \rho_A(a)=n_{A}^{-1}\sum_{i=1}^{n_{A}}\Char{\rand{a}_{i,A}=a}$ for any $a\in A$. 
\par
A natural test statistic based on Theorem~\ref{thm:main theorem} is
\begin{align*}
    \mathrm{T}_n=n\min_{v\in\Real^{\abs{R_{\Pi}}}_{+}}\norm{\hat \rho-Bv}^2,
\end{align*}
where $n=\sum_{A}n_{A}$ is the sample size.
\par
Let $\hat{\rho}^{*}_l$, $l=1,\dots,L$, be bootstrap replications of $\hat \rho$. Let  $\tau_n\geq 0$ be a tuning parameter and $\iota$ be a vector of ones of dimension $\abs{R_{\Pi}}$.\footnote{We conducted tests with $\tau_n=\sqrt{\dfrac{\log(\min_{A}n_{A})}{\min_{A}n_{A}}}$ following KS.} To compute the critical values of $\mathrm{T}_n$ we follow the bootstrap procedure proposed in KS:
\begin{enumerate}
    \item Compute $\hat\eta_{\tau_n}=Bv_{\tau_n}$, where $v_{\tau_n}$ solves
    \[
    n\min_{[v-\tau_n \iota/{\abs{R_{\Pi}}}]\in\Real^{\abs{R_{\Pi}}}_{+}}\norm{\hat{\rho}-Bv}^2;
    \]
\item Compute the bootstrap test statistic
\[
\mathrm{T}^{*}_{n,l}=n\min_{[v-\tau_n \iota/d]\in\Real^{\abs{R_{\Pi}}}_{+}}\norm{\hat{\rho}^{*}_l-\hat{\rho}+\hat\eta_{\tau_n}-Bv}^2,\quad l=1,\dots,L;
\]
\item Use the empirical distribution of the bootstrap statistic to compute critical values of $\mathrm{T}_n$.
\end{enumerate}
For a given confidence level $\alpha\in (0,1/2)$, the decision rule for the test is ``reject the null hypothesis if $\mathrm{T}_n>\hat c_{1-\alpha}$'', where $\hat c_{1-\alpha}$ is an $(1-\alpha)$-quantile of the empirical distribution of the bootstrap statistic. 

\section{Testing for RRDEU and REU in Experimental Data}\label{secc:test}
\subsection*{Data}
To illustrate the proposed methodology we use a subsample of the dataset from \citet{ABKK2021} to test for consistency of the behavior of a population of DMs with RRDEU and REU. Crucially, we use only one treatment (frame) of \citet{ABKK2021} corresponding to a situation where the cost of experimental subjects to pay attention is designed to be low. In that sense, we can avoid thinking about limited attention/consideration in this paper. 
\par 
\citet{ABKK2021} conducted the experiment in Amazon MTurk for a large cross-section with at most two (disjoint) choice sets per individual. In the experiment, the set of prizes is $X=\{0,10,12,20,30,48,50\}$. The set of lotteries $\Pi$ is presented in Table~\ref{lotteries}. For example, the first lottery can be written as $l_1=\left(1/2,0,0,0,0,0,1/2\right)\tr$ (i.e. it assigns positive probability to prizes $0$ and $50$ only). Similarly, the second lottery can be written as $l_2=\left(0,0,1/2,0,0,1/2,0\right)\tr$. 

\begin{table}[ht]	\caption[Lotteries]{\textsc{Lotteries measured in tokens, expected values, and variance}}\label{lotteries}  \centering
\begin{tabular}{@{\extracolsep{1pt}}|c l| c | c |@{}}
\cline{1-4}
\multicolumn{2}{|c|}{} &  &  \\	[.25ex] 
\multicolumn{2}{|c|}{\textsc{Lottery}} & \textsc{Expectation} & \textsc{Variance}  \\ 
 \cline{1-4} 
 &&&\\ [.1ex]
(1)&$\frac{1}{2} 50 + \frac{1}{2} 0$ & 25 & 625\\[1ex]
(2)&$\frac{1}{2} 30 + \frac{1}{2} 10$ & 20 & 100\\[1ex]
(3)&$\frac{1}{4} 50 + \frac{1}{4} 30 + \frac{1}{4} 10 + \frac{1}{4} 0$ & 22.5 & 368.75\\[1ex]
(4)&$\frac{1}{4} 50 + \frac{1}{5} 48 + \frac{3}{20} 14 + \frac{2}{5} 0$ & 24.125 & 511.73\\[1ex]
(5)&$\frac{1}{5} 48 + \frac{1}{4} 30 + \frac{3}{20} 14 + \frac{1}{4} 10+ \frac{3}{20} 0$ & 21.625 & 251.11\\[1ex]
(o) & 12 with probability 1  & 12 & 0\\ \cline{1-4}\end{tabular}
\end{table}

All sessions were run between August 25, 2018 and September 17, 2018 on the MTurk platform with surveys designed in Qualtrics. The data contains choices 4099 choices from all possible nonsingleton menus that contain the default lottery $o$, which pays 12 tokens with certainty. For more details on the experiment see \citet{ABKK2021}.
\par
\paragraph{Structure of the Lotteries.} 
Here we show the special structure of our alternatives that allows us to test EU (and the independence axiom) as a restriction on the set of linear orders (i.e, we use $R^{EU}_{\Pi}$ to denote the restriction of the set of linear orders consistent with EU to $\Pi$).  
The experiment was design with power against REU. Here we will use the \citet{ABKK2021} experiment to test also RRDEU.
\par
If $\succ\in R^{EU}_{\Pi}$, then independence implies that for any $p,q,r\in\Delta(X)$ and any $\alpha\in(0,1)$
\[
p\succ q\iff \alpha p +(1-\alpha)r \succ \alpha q +(1-\alpha) r\,.
\]

To understand additional restrictions imposed by independence, define the auxiliary lottery $r=\left(0,2/5,0,3/10,0,0,3/10\right)\tr$, and note the following relations among lotteries in $\Pi$: 
\begin{equation*}
l_3=\frac{1}{2}l_1+\frac{1}{2}l_2,\quad l_4=\frac{1}{2}l_1+\frac{1}{2}a,\quad
l_5=\frac{1}{2}l_2+\frac{1}{2}a\,.
\end{equation*}
This structure restricts the possible orders that are compatible with expected utility: (i) if $l_1\succ l_2$, then $l_1\succ l_3$, $l_3\succ l_2$, and $l_4\succ l_5$; or (ii) if $l_2\succ l_1$, then $l_2\succ l_3$, $l_3\succ l_1$, and $l_5\succ l_4$. 
\par
However, the previous restrictions are only implications of the expected utility assumption. The necessary and sufficient conditions have been spelled out in Theorem~\ref{thm:main theorem}.

\subsection*{Results}
We apply the procedure described in Section~\ref{subsec: econometric testing} to test whether REU and RRDEU can explain the data. The results of testing are presented in Table~\ref{results_joit_pref}. In this table, we report the values of the test statistic and the corresponding p-values coming from the bootstrap distribution ($1000$ bootstrap replications for every test statistic were conducted) for two models.\footnote{The p-value is interpreted as the probability of observing a realization of the test statistic that is above the one that is actually observed due to sample variability, if the null hypothesis is indeed correct. Then, the smaller the p-value is, the more evidence the researcher has to reject the hypothesis of the validity of a given model.}
We reject expected utility at the $5$ percent significance level. At the same time, we cannot reject the rank-dependent expected utility model at any standard significance level.
\par 
We must highlight that the fact that RRDEU is not rejected while REU is, is not entirely surprising since RRDEU is a strict generalization of REU. However, there is no reason \emph{a priori} to think that RRDEU explains the experimental datasets used here. In that sense, we add some robust nonparametric empirical evidence supporting the use of RRDEU instead of REU to explain choice over risky prospects.

\begin{table}
\begin{center} 
\caption[Consistency preferences]{\textsc{Testing Results}  } \label{results_joit_pref} 
\centering
\begin{tabular}{@{\extracolsep{2pt}} |l cc  |  @{}}
\cline{1-3}
\multicolumn{3}{|c|}{}\\
\textsc{Model}&\textsc{$\mathrm{T}_n$}&\textsc{p-value}\\ [2ex]
\cline{1-3} 
\multicolumn{3}{|c|}{}\\
REU &	387.72 & 0.013\\[2ex]
RRDEU &130.75 & 0.906\\[2ex]
\multicolumn{3}{|c|}{}\\ \cline{1-3}
\multicolumn{3}{l}{\footnotesize{Notes: Number of bootstrap replications=1000.}}
\end{tabular}
\end{center} 
\end{table}

\section{Conclusions}\label{sec: conclusion}
We have proposed a new characterization of RDEU when we observe a cross-section of choices from heterogeneous DMs that choose from a finite set of lotteries and a finite collection of menus. We have established a nonparametric and computationally feasible test of RRDEU (and as a byproduct a test of REU). Our main result lends itself to statistical testing using the tools in KS. This test is wide applicable in experimental and field stochastic choice datasets. 

\bibliographystyle{econ}
\phantomsection\addcontentsline{toc}{section}{\refname}\bibliography{rankdepEU}
\end{document}